\documentclass[journal]{IEEEtran}
\renewcommand{\baselinestretch}{0.975}
\usepackage{amsmath,amssymb,amsfonts}
\usepackage{amsthm}
\usepackage{algorithmic}
\usepackage{algorithm}
\usepackage{array}
\usepackage{textcomp}
\usepackage{stfloats}
\usepackage{url}
\usepackage{verbatim}
\usepackage{graphicx}
\usepackage{subcaption}
\usepackage{cite}
\usepackage{fancyhdr}

\usepackage[letterpaper, left=1in, right=1.1in, top=0.75in, bottom=1.1in]{geometry}

\usepackage{colortbl}

\newcommand{\remove}[1]{}
\newtheorem{theorem}{Theorem}

\newcommand\blfootnote[1]{%
  \begingroup
  \renewcommand\thefootnote{}\footnote{#1}%
  \addtocounter{footnote}{-1}%
  \endgroup
}

\fancypagestyle{firstpage}
{
    \fancyhead[C]{Accepted by 2024 IEEE International Workshop on Computer Aided Modeling and Design of Communication Links and Networks (CAMAD), \textcopyright 2024 IEEE}
    
}

\begin{document}

\title{Optimal Digital Twinning of Random Systems\\ with Twinning Rate Constraints}

\author{\IEEEauthorblockN{Caglar Tunc}\\
\IEEEauthorblockA{Sabancı University, İstanbul, Türkiye\\caglar.tunc@sabanciuniv.edu}

\vspace{-9mm}
}





\maketitle
\begin{abstract}
With the massive advancements in processing power, Digital Twins (DTs) have become powerful tools to monitor and analyze physical entities. However, due to the potentially very high number of Physical Systems (PSs) to be tracked and emulated, for instance, in a factory environment or an Internet of Things (IoT) network, continuous twinning might become infeasible. In this paper, a DT system is investigated with a set of random PSs, where the twinning rate is limited due to resource constraints. Three cost functions are considered to quantify and penalize the twinning delay. For these cost functions, the optimal twinning problem under twinning rate constraints is formulated. In a numerical example, the proposed cost functions are evaluated for two, one push-based and one pull-based, benchmark twinning policies. The proposed methodology is the first to investigate the optimal twinning problem with random PSs and twinning rate constraints, and serves as a guideline for real-world implementations on how frequently PSs should be twinned.\blfootnote{This work was initially submitted when Caglar Tunc was with 6GEN Lab, Turkcell, İstanbul, Türkiye.}
\end{abstract}
\thispagestyle{firstpage}
\vspace{-3mm}
\section{Introduction}
Digital Twins (DTs) have emerged as a transformative technology, which are used to monitor, analyze and optimize physical entities. By creating virtual replicas of physical systems (PSs), DTs enable real-time data collection, simulation, and optimization, thus providing useful insights and control over complex processes. The concept of DTs, initially adopted by the manufacturing and aerospace industries, has now covered various sectors, including healthcare, urban planning, and, more importantly, the emerging field of 6G communications \cite{khan_2022}. DTs play even a more crucial role in the development of 6G networks, envisioned as a step beyond the current 5G infrastructure. DT-enabled 6G networks promise enhanced network management, predictive maintenance, and intelligent resource allocation, empowering a seamless integration of physical and digital worlds. This integration has the potential to address the increasing demand for high-speed, low-latency communications essential for applications such as autonomous driving, remote surgery, and immersive augmented reality use cases.

Despite their potential, the implementation of DTs in large-scale environments, such as factories or Internet of Things (IoT) networks, presents significant challenges. Continuous twinning of a vast number of PSs, as required for real-time monitoring and control, is often infeasible due to resource constraints. This limitation necessitates innovative strategies to optimize the twinning process, ensuring that the benefits of DTs are realized without overwhelming the available computational resources \cite{masaracchia_2022}.

The optimization of DT system synchronization under resource constraints is a critical research area, particularly in the context of 6G communications. Authors of \cite{kuruvatti_2022} provide a comprehensive survey on use cases and application areas for 6G systems with DT technology, underscoring the need for efficient algorithms and architectures to manage the high volume of data and the dynamic nature of PSs. Reference \cite{fuller_2020} discusses the enabling technologies and challenges associated with DTs, emphasizing the importance of addressing issues such as data security, interoperability, and scalability.

In industrial IoT applications, the deployment of DTs must consider the specific requirements and constraints of the environment. A higher sampling/twinning frequency means a better performance for the DT and a more accurate representation of the PS, which is a critical factor in maintaining system performance and reliability \cite{xu_2023}. However, this can be limited due to resource and energy constraints, which needs to be addressed carefully while designing the DT system with data collection pipelines.

In this paper, we investigate a DT system that serves as a digital replica of a network of PSs. The goal of the DT system is to have an accurate and recent representation of the PSs, which is achieved by sampling/twinning them as frequently as possible. Each PS is a random process by its nature, and twinning the PSs is critical for capturing the state of these random processes. We are interested in the case where the twinning rate is limited due to resource constraints; hence, the DT system might fail to capture all of the changes occurring in the PSs. We also assume that there is a synchronization delay between the DT and the corresponding PS after the twinning is initiated, due to data collection and processing times. To quantify the freshness of the DT system, we define three cost functions, each of which can be useful and appropriate for different systems, depending on the particular application area.

Our main contributions are listed as follows:
\begin{itemize}
    \item To the best of our knowledge, our work is the first to formulate an optimal twinning problem with an overall twinning rate constraint.
    \item We propose three different cost functions to quantify the freshness of a DT system that emulates a set of random PS. Each cost function corresponds to a specific way of penalizing the mismatch between the PSs and the corresponding DT system.
    \item For the defined cost functions, we formulate an optimization problem to compute optimal twinning policies that aim to minimize the long-run average cost value, given a twinning rate constraint.
\end{itemize}

\vspace{-2mm}
\section{Related Work}

The tradeoff between energy/resource consumption and twinning/sampling rate have attracted attention in the literature \cite{cakir_2023,cakir_2024,duran_2023_2,duran_2023_1,vaezi_2023,li_2024_1, aghaei_2023}. The concept of adaptive twinning was investigated in \cite{cakir_2023}, where a so-called twin alignment ratio is defined the synchronization performance of the DT is gauged according to the defined metric. In the related study \cite{cakir_2024}, authors investigate adaptive twinning for dynamic PSs under energy and resource constraints and propose a reinforcement learning (RL)-based adaptive DT model for energy management in green cities. Similarly, the study in \cite{duran_2023_2} explores mechanisms and architecture for efficient and smart data collection to avoid collection and improve DT performance. As for the timeliness and update frequency of DTs, the concept of Age of Twin (AoT) was introduced in \cite{duran_2023_1}, which underscores the importance of timely DT updates in the 6G ecosystem to ensure accurate and reliable representations of PSs. Studies in \cite{vaezi_2023} and \cite{li_2024_1} further emphasize the importance of minimizing application request delays and maintaining fresh information, respectively, to enhance user satisfaction and system performance in DT-empowered edge computing environments. Aghaei and Zhao \cite{aghaei_2023} focus on optimizing DT response times for time-sensitive applications with a Markov Decision Process (MDP) formulation in order to ensure timely and efficient DT operation. In a related line of research, authors of \cite{akar_2024} derive the optimal sampling rates for a continuous-time Markov chain (CTMC) system under total average sampling rate constraint. On the other hand, none of these studies have examined the cost-optimal digital twinning policies under twinning rate constraints, which is the main focus of this paper.

Optimization strategies for DT placement and resource allocation have also been studied \cite{gu_2024,zhang_2024,li_2024_2,lu_2021,liu_2024}. Intelligent placement algorithms for wireless DT networks using bandit learning are developed in \cite{gu_2024}, aiming to optimize DT placement dynamically. Similarly, the study in \cite{zhang_2024} focuses on cost minimization for DT placements in mobile edge computing, providing strategies to reduce latency and improve efficiency. \cite{li_2024_2} addresses mobility-aware utility maximization in DT-enabled serverless edge computing, proposing methods to optimize resource allocation in dynamic environments. Lu et al. \cite{lu_2021} investigate adaptive edge association for wireless DT networks in 6G, and propose mechanisms to enhance network performance by dynamically associating edge resources. Another study that focuses on intelligent service deployment \cite{liu_2024} explores edge-cloud collaboration for service deployment and resource allocation in DT-enabled 6G networks to balance the load between edge and cloud resources.
\vspace{-2mm}
\section{System Model}
We consider the system in Figure~\ref{fig:dt_diagram} with $K$ PSs and $K$ corresponding DTs. Each PS in a network is monitored and modeled by a corresponding DT. Due to the random nature of the environment that is modeled, physical entities that are correlated with each other are grouped into a PS, which is assumed to be evolving independently from other PSs in the network. On the other hand, we consider the scenario where with each query initiated by the monitor side, all $K$ PS are sampled/twinned simultaneously, and there is no option for individual PS sampling/twinning. Moreover, there is a fixed synchronization delay of $\Delta$ time units after the query is sent to the PS network. In other words, the most recent DT representation of a PS has an age of $\Delta$ time units. After a DT is synchronized with its corresponding PS, its state remains the same until a new synchronization query is sent to the PS network.

\begin{figure}[!t]
\centering
\includegraphics[width=2.8in]{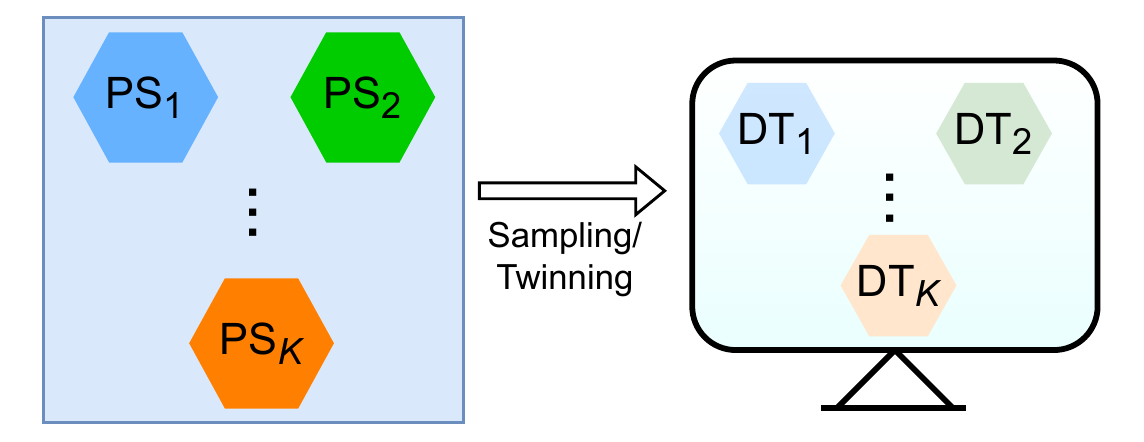}
\caption{Monitoring and twinning of the physical systems.}
\vspace{-6mm}
\label{fig:dt_diagram}
\end{figure}

Each PS is assumed to be a Markovian process with the state at time $t$ denoted by $S_i(t)$, for $i\in\{1,2,...,K\}$. The state of the PS network at time $t$ is denoted by $K$-tuple $S(t)=(S_1(t), S_2(t), ..., S_K(t))$. Similarly, we denote the state of $i$-th DT at time $t$ by $\hat{S}_i(t)$, for $k\in\{1,2,...,K\}$. The state of the DT system at time $t$ is denoted by the $K$-tuple $\hat{S}(t) = (\hat{S}_1(t), \hat{S}_2(t), ..., \hat{S}_K(t))$.

We assume a fixed synchronization delay between the PS and its corresponding DT, which is denoted by $\Delta$. With $t_0$ being the time at which the twinning query is initiated, we have the following relation:
\begin{equation}
    \hat{S}_i(t_0+\Delta) = S_i(t_0),
\end{equation}
for $i\in\{1,2,...,K\}$. The following relation holds for the state of the PS network and the DT system:
\begin{equation}
    \hat{S}(t_0+\Delta) = S(t_0).
\end{equation}
Hence, the mismatch between the PS network and the DT system can be caused by (i) synchronization delay $\Delta$, and (ii) state transitions in the PS network which are not yet twinned.

To quantify the twinning and synchronization performance, we propose to use three cost functions. For this, we first define $T_1^{(i)}$ as the first transition time for PS $i$ out of the sampled state $S_i(t_0)$, sampled at time $t_0$. Similarly, let $T_1$ denote the first transition time for PS network out of the sampled state tuple $S(t_0)=(S_1(t_0), S_2(t_0), ..., S_K(t_0))$. With these definitions, we list these cost functions as follows:
\begin{align}
    C_1(t) = \mathbf{1}(T_1<t+\Delta)\label{eq:C1}\\
    C_2(t) = \sum_{i=1}^K w_i\mathbf{1}(T^{(i)}_1<t+\Delta)\label{eq:C2}\\
    C_3(t) = f(S(t),\hat{S}(t))\label{eq:C3},
\end{align}
where $C_j(t)$ denotes the cost value at time $t$ for cost function $j$, for $j\in\{1,2,3\}$. We explain and provide motivation for using each cost function as follows:
\begin{itemize}
    \item $C_1(t)$ tracks the system state, and whenever an individual PS transitions out of its sampled state, it takes the value 1. Even if the PS returns to its sampled state after leaving it, we assume that the cost is still 1 since at least two PS state transitions have been missed by the DT system.
    \item $C_2(t)$ is similar to weighted Hamming distance, with only difference being that once the state of a PS changes, the cost becomes non-negative until next twinning incident, while for Hamming distance, it becomes when the PS returns to the sampled state. In $C_2(t)$, each PS is weighed differently while calculating the total cost. This cost function can fit well to scenarios where different PSs in the network have different importance, and missing the state transitions occurring in these states should be penalized more harshly than others.
    \item $C_3(t)$ can take an arbitrary function to quantify the difference the states of two processes. Some commonly used examples for the distance function $f(\cdot)$ in (\ref{eq:C3}) are generalized/weighted Hamming distance, Euclidean distance, Manhattan distance, Chebyshev distance, and cosine similarity \cite{aggarwal_2013}. Depending on the system in hand, any distance function can be utilized to gauge how much the system has iterated since the last twinning instance.
\end{itemize}

\begin{figure}[!t]
\centering
\includegraphics[width=2.6in]{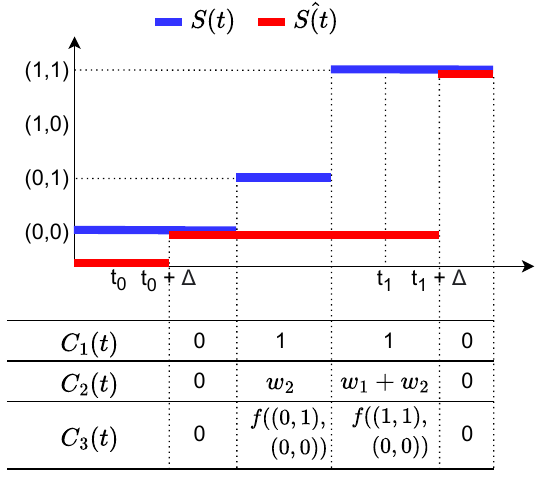}
\caption{Sample path for the states of the PSs and DTs, and the corresponding values of the cost functions.}
\vspace{-5mm}
\label{fig:sample_path}
\end{figure}

We illustrate a sample path of $S(t)$ and $\hat{S}(t)$ with $K=2$ and $S_i\in\{0,1\}$, for $i=\{1,2\}$ in Figure~\ref{fig:sample_path}. $t_0$ and $t_1$ denote the first and second twinning instances, respectively. With the first twinning instance, $\hat{S}(t)$ sets its value to $S(t_0)$ at time $t=t_0 + \Delta$, which is the twinning instance plus the synchronization delay. Until the first transition of $S(t)$ from $(0,0)$ to $(0,1)$, all cost functions are set to 0. When the transition occurs, $C_1(t)$, $C_2(t)$, and $C_3(t)$ are calculated according to (\ref{eq:C1}), (\ref{eq:C2}), and (\ref{eq:C3}), respectively. Similarly, with the transition of $S(t)$ from $(0,1)$ to $(1,1)$, cost values are updated. Finally, with the second twinning instance, after the synchronization delay, the DT system is synchronized and the cost values are set to 0.

For a given PS network and associated cost function, it is extremely useful to determine how the cost grows in time, especially when the twinning rate is limited due to resource constraints. Therefore, next, we analyze the expected values of each cost function. The expected value of cost function $C_1(t)$ is given by the following theorem.

\vspace{-2mm}
\begin{theorem}
\label{th:theorem1}
Assuming an initial twinning instance of $t_0=0$, expected value of $C_1(t)$ is given by the following expression:
\begin{equation}
    \mathbb{E}[C_1(t)] = 1-\prod_{i=1}^K \sum_k \pi_{ik} (e^{Q_i(t+\Delta)})_{kk},
\end{equation}
where $\mathbb{E}[\cdot]$ denotes the expectation operator, $\pi_{ik}$ is the steady-state probability of state $k$ for PS $i$, $Q_i$ is the transition rate matrix for PS $i$, and $(e^{(\cdot)})_{kk}$ is the entry at the $k$-th row and $k$-th column of the exponential matrix $e^{(\cdot)}$, for $i\in\{1,2,.,,,.K\}$.
\end{theorem}
\begin{proof}
    Taking the expected value of both sides of (\ref{eq:C1}) yields:
    \begin{align}
        \mathbb{E}[C_1(t)] &= \mathbb{E}[\mathbf{1}(T_1<t+\Delta)]\\
        & = P(T_1<t+\Delta)\\
        & = 1-P(T_1\geq t+\Delta)\label{eq:9}\\
        & = 1- \prod_{i=1}^K P(T^{(i)}_1\geq t+\Delta)\label{eq:10}\\
        & = 1- \prod_{i=1}^K \sum_k \pi_{ik} P_{kk}(t+\Delta)\label{eq:11}\\
        & = 1- \prod_{i=1}^K \sum_k \pi_{ik} (e^{Q_i(t+\Delta)})_{kk}\label{eq:12}.
    \end{align}
    $P_{kk}(t)$ denotes the probability that the system remains in state $k$ $t$ time units after initialled sampled at time $t_0=0$. (\ref{eq:10}) from (\ref{eq:9}) follows from the assumption that PSs are independent, and (\ref{eq:12}) from (\ref{eq:11}) from \cite{gagniuc_2017}, which completes the proof.
\end{proof}

Similarly for $C_2(t)$, we use the following theorem to compute the expected value.
\begin{theorem}
Assuming an initial twinning instance of $t_0=0$, expected value of $C_2(t)$ is given by the following expression:
\begin{equation}
    \mathbb{E}[C_2(t)] = \sum_{i=1}^{K} w_i\left(1-\sum_k \pi_{ik} (e^{Q_i(t+\Delta)})_{kk}\right).
\end{equation}
\end{theorem}
\begin{proof}
    Taking the expected value of both sides of (\ref{eq:C1}) yields:
    \begin{align}
        \mathbb{E}[C_2(t)] &= \mathbb{E}\left[\sum_{i=1}^K w_i\mathbf{1}(T^{(i)}_1<t+\Delta)\right]\label{eq:14}\\
        &= \sum_{i=1}^K w_i\mathbb{E}\left[\mathbf{1}(T^{(i)}_1<t+\Delta)\right]\label{eq:15}\\
        &= \sum_{i=1}^K w_i\left(1-P(T^{(i)}_1\geq t+\Delta)\right)\label{eq:16}\\
        &= \sum_{i=1}^K w_i\left(1-\sum_k \pi_{ik} (e^{Q_i(t+\Delta)})_{kk}\right)\label{eq:17},
    \end{align}
    where we used the independence of PSs going from (\ref{eq:14}) to (\ref{eq:15}), and similar arguments as in the proof of Theorem~\ref{th:theorem1} going from (\ref{eq:16}) to (\ref{eq:17}).
\end{proof}

The expected cost value of the cost function $C_3(t)$ depends on the particular choice of distance function $f(\cdot)$. As an example, we have derived the expected cost value for the distance function $f(\cdot)$ being a slight variation of the generalized Hamming distance in Theorem 2. We leave the analysis of other distance functions to future work.

\section{Optimal Monitoring Problem Formulation}
Let $\lambda_i$ denote the $i$-th twinning instance. We denote the average twinning rate until time $t$ by $\lambda(t)$, which is computed as follows:
\begin{equation}
    \lambda(t) = \frac{1}{t}\sum_{i=1}^{\infty} \mathbf{1}(\lambda_i \leq t).
\end{equation}
Similarly, we calculate the long-run average twinning rate, denoted by $\lambda_{avg}$, as follows:
\begin{equation}
    \lambda_{avg} = \lim_{T\rightarrow \infty}\frac{1}{T}\sum_{i=1}^{\infty} \mathbf{1}(\lambda_i \leq T).
\end{equation}

To formulate the optimal twinning policy problem, for a chosen cost function $C_i(t)$ for $i\in\{1,2,3\}$, we denote the expected cost when twinning policy $\pi$ is followed by $\mathbb{E}^{\pi}[C_i(t)]$. We have the following optimization problem that gives the optimal twinning policy:

\begin{subequations}
\begin{alignat}{2}
&\! \underset{\mathcal{\pi\in\mathcal{P}}}{\text{minimize}}     &\qquad& \mathbb{E}^{\pi}[C_i(t)]\label{eq:objective}\\
&\text{subject to} &      & \lambda_{avg}\leq \Lambda,\label{eq:constraint1}
\end{alignat}
\end{subequations}
where $\mathcal{P}$ denotes the set of feasible twinning policies, and $\Lambda$ denotes the constraint on the average twinning rate.

Finding the optimal twinning policy requires solution of (\ref{eq:objective}), given the underlying distributions for the PSs. If the infinitesimal generator matrices of PSs 1 through $K$ are given, then the entire system can be formulated as an MDP, and the optimal policy can be derived by using standard algorithms such as value-\cite{tunc_2019} of policy-iteration \cite{zhou_2019}. In the numerical example, instead of solving for the optimal policy, we focus on two sampling policies that guarantee that the twinning rate constraint (\ref{eq:constraint1}) is satisfied: (i) pull-based random twinning policy (PRTP), which samples the PS network based on random twinning queries generated at the monitor side, and (ii) push-based probabilistic twinning policy (PPTP), which is initiated at the PS network side depending on the state transitions of the PS network.

\section{Numerical Example}
To illustrate the defined cost functions and two benchmark twinning policies, we consider a simple scenario with $K=2$ PSs. For the infinitesimal generator matrices $Q_1$ and $Q_2$, we set the following values:

\begin{equation}
    Q_1=\begin{bmatrix}
  -1 & 1\\
  2 & -2
\end{bmatrix},
Q_2\begin{bmatrix}
  -3 & 3\\
  6 & -6
\end{bmatrix}.
\end{equation}
Moreover, we set the weight matrix for cost function $C_2(t)$ as $w = [5,1]$. As cost function $C_3(t)$, we consider the Euclidean distance function, which is expressed as follows:
\begin{equation}
    f\left(S(t),\hat{S}(t)\right) = \sum_{i=1}^K \sqrt{\left(S_i(t)-\hat{S}_i(t)\right)^2}.
\end{equation}

The first twinning policy we investigate is PRTP, which is triggered on the monitor side according to an exponential random variable with rate $\lambda_{avg}$. Therefore, it is guaranteed that on average, twinning rate will converge to $\lambda_{avg}$. The second policy we use as a benchmark is PPTP, which is triggered probabilistically on the PS network side with each state transition. To guarantee converging to $\lambda_{avg}$ on the long-run, we set the probability of twinning with each state transition as follows:
\begin{equation}
    p_t = \min(1,\lambda_{avg}/\sigma),
\end{equation}
where $\sigma$ is the total transition rate for both PSs, which is calculated as follows:
\begin{equation}
    \sigma = \pi_{11}r_{11} + \pi_{12}r_{12} + \pi_{21}r_{21} + \pi_{22}r_{22},
\end{equation}
where $r_{ij}$ denote the total transition rate out of state $j$ for PS $i$. For this particular example, solving for $\pi_{ij}$ yields $\pi_{11}=\pi_{21}=2/3$ and $\pi_{12}=\pi_{22}=1/3$, which is simply derived from the steady-state solution of a two-state Markov chain. Hence, $\sigma=4/3+4=16/3$, and $p_t=3/16$.

We vary the synchronization delay $\Delta$ between 0 (i.e., instantaneous synchronization) and 0.6, and the twinning rate between 1 and 30. Average values of the cost functions $C_1(t)$, $C_2(t)$, and $C_3(t)$, denoted by the expectation operator $\mathbb{E}[\cdot]$, are illustrated in subplots (a), (b), and (c), respectively, of Figure~\ref{fig:ex1} and \ref{fig:ex2} for PRTP and PPTP, respectively. The first observation we make is that PPTP outperforms PRTP in general, since it is a push-based policy that exploits the knowledge of state transitions of the PS network. Another observation is that as the synchronization delay $\Delta$ gets larger, twinning actions might fail to synchronize the DT system with the PS network. This is because the probability that there is a state transition within the PS network during the synchronization gets larger, and hence, the synchronization for the latest twinning incident fails. One interesting behavior we observe from Figures~\ref{fig:ex1} and \ref{fig:ex2} is that, the average cost values do not always decrease monotonically with the increasing twinning rate $\lambda_{avg}$. The main reason behind this phenomenon is that, especially when the synchronization delay is high, the adopted twinning policy should wait until the most recent synchronization completes. Otherwise, the synchronization performance of the DT system can actually be deteriorated. This underscores the need for analyzing the system constraints and designing intelligent twinning policies according to the random system and constraints in hand.

\begin{figure*}[!t]
\centering
\includegraphics[width=6.3in]{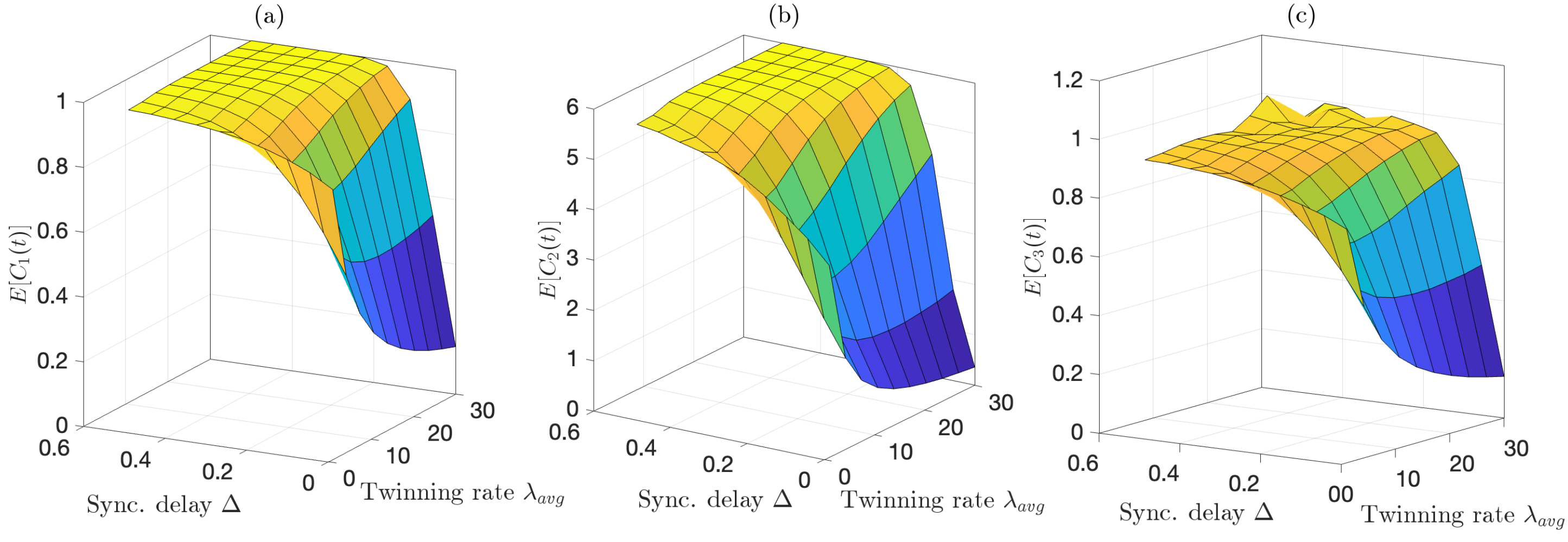}
\caption{Average values for the cost functions (a) $C_1(t)$, (b) $C_2(t)$, and (c) $C_3(t)$, for varying values of synchronization delay $\Delta$ and average twinning rate $\lambda_{avg}$, for PRTP.}
\label{fig:ex1}
\end{figure*}
\begin{figure*}[!t]
\centering
\includegraphics[width=6.3in]{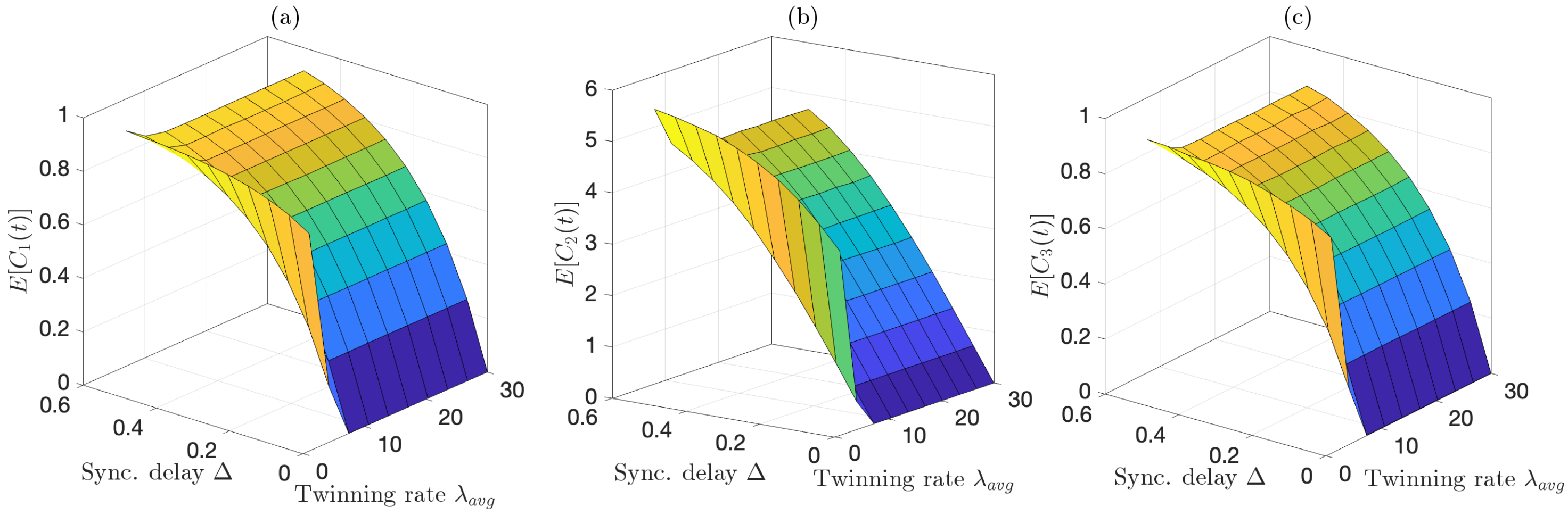}
\caption{Average values for the cost functions (a) $C_1(t)$, (b) $C_2(t)$, and (c) $C_3(t)$, for varying values of synchronization delay $\Delta$ and average twinning rate $\lambda_{avg}$, for PPTP.}
\label{fig:ex2}
\end{figure*}

\section{Conclusion}
In this paper, we considered a Digital Twin (DT) system, modeling/emulating a Physical System (PS) network under twinning rate constraints. By defining three cost functions to quantify and penalize twinning delays and model mismatch between the PSs and DTs, we formulated the optimal twinning problem and outlined a framework for determining twinning strategies for the PSs. Our methodology introduces a novel approach to tackle the optimal twinning problem with random PSs and twinning rate constraints, offering practical guidelines for real-world DT deployments.

Our approach balances the benefits of continuous twinning with resource limitations, offering a practical solution adaptable to industrial and IoT environments. Future research will explore specific use cases and AI integration for determining the optimal twinning policies, and develop adaptive twinning algorithms to enhance DT systems' capabilities.

This study provides a foundational framework for optimizing DT systems under resource constraints, bridging the gap between theoretical research and practical implementation, and paving the way for more efficient industrial processes and communication networks. Future research will focus on solving for the optimal twinning policy for a given set of random distribution for the PSs. When the dynamics of these distributions are known, MDP formulation will be used to obtain the optimal policies. Otherwise, RL will be utilized to learn the optimal twinning policies that minimize the desired cost function.

\bibliographystyle{IEEEtran}
\bibliography{IEEEabrv,whole}

\end{document}